\documentclass{vldb}

\usepackage{amsmath,amsfonts,graphpap,amscd,mathrsfs,graphicx,lscape}
\usepackage{epsfig,amssymb,amstext,xspace,bbm}
\usepackage{theorem,dsfont}

\usepackage[ruled,vlined]{algorithm2e}
\renewcommand{\algorithmcfname}{ALGORITHM}
\SetAlFnt{\small}
\SetAlCapFnt{\small}
\SetAlCapNameFnt{\small}
\SetAlCapHSkip{0pt}
\IncMargin{+1em}

\usepackage{color}              % Need the color package
\usepackage{epsfig}

\newtheorem{theorem}{Theorem}

\newtheorem{defn}[theorem]{Definition}

{\theorembodyfont{\rmfamily} }
\def\squareforqed{\hbox{\rule{2.5mm}{2.5mm}}}
\def\blackslug{\rule{2.5mm}{2.5mm}}
\def\qed{\hfill\blackslug}

\def\QED{\ifmmode\squareforqed % in mathmode : print just the square
  \else{\nobreak\hfil   % \hfil to end of current line
    \penalty50                 % penalty 50 for breaking the line here
    \hskip1em                  % leave at least 1em before the square
    \null                      % \hbox{}
    \nobreak                   % prohibit line break
    \hfil                      % another \hfil (if a break occurred)
    \squareforqed              % put the square here
    \parfillskip=0pt           % the line really ends here
    \finalhyphendemerits=0     % ignore a hyphen on the line above
    \endgraf}                  % end the paragraph
  \fi}

\def\blksquare{\rule{2mm}{2mm}}
\def\qedsymbol{\blksquare}
\newcommand{\bg}[1]{\medskip\noindent{\bf #1}}
\newcommand{\ed}{{\hfill\qedsymbol}\medskip}

\newenvironment{example}{\bg{Example. }}{\ed}

\newcommand{\comment}[1]{}

\newcommand{\R}{\ensuremath{\mathbb R}}

\newcommand{\qset}{\ensuremath{\mathcal{Q}}}	
\newcommand{\bqset}{\ensuremath{\mathcal{B}}}
\newcommand{\dqset}{\ensuremath{\mathcal{D}}}
\newcommand{\bq}{\ensuremath{b}}
\newcommand{\dq}{\ensuremath{q}}
\newcommand{\rel}{\ensuremath{\mathcal{R}}}
\newcommand{\bpr}[1]{\ensuremath{p_{#1}}}
\newcommand{\bprs}{\ensuremath{{\bf p_{\bqset}}}}
\newcommand{\pr}[1]{\ensuremath{p(#1)}}
\newcommand{\prp}[1]{\ensuremath{\hat{p}(#1)}}
\newcommand{\orac}[1]{\ensuremath{\mathcal{A}_{D}^{\bqset,\bprs}(#1)}}
\newcommand{\rev}{\ensuremath{\mathcal{R}}}

\begin{document}

\title{Pricing Queries (Approximately) Optimally}
\numberofauthors{2}
\author{
\alignauthor
Vasilis Syrgkanis\\
\affaddr{Microsoft Research, NYC}\\
\email{vasy@microsoft.com}
\alignauthor
Johannes Gehrke\\
\affaddr{Microsoft}\\
\email{johannes@microsoft.com}
}
% \author{Vasilis Syrgkanis\\
% \affaddr{Cornell University}\\
%        \email{vasilis@cs.cornell.edu}
% \and
% Johannes Gehrke\\
% \affaddr{Cornell University}\\
%        \email{johannes@cs.cornell.edu}
% }
\maketitle

\begin{abstract}
Data as a commodity has always been purchased and sold. Recently, web services that are data marketplaces have emerged that match data buyers with data sellers. So far there are no guidelines how to price queries against a database. We consider the recently proposed query-based pricing framework of Koutris et al.~\cite{Koutrisb} and ask the question of computing optimal input prices in this framework
by formulating a buyer utility model.

We establish the interesting and deep equivalence between arbitrage-freeness in the query-pricing framework and envy-freeness in pricing theory for appropriately chosen buyer valuations. Given the
approximation hardness results from envy-free pricing we then develop logarithmic approximation pricing algorithms exploiting the max flow interpretation of the arbitrage-free pricing for the restricted query
language proposed by \cite{Koutrisb}. We propose a novel polynomial-time logarithmic approximation pricing scheme and show that our new scheme performs better than the existing envy-free pricing
algorithms instance-by-instance. We also present a faster pricing algorithm that is always greater than the existing solutions, but worse than our previous scheme.  We experimentally show how our pricing algorithms perform with respect to the existing envy-free pricing algorithms and to the optimal exponentially computable solution, and our experiments show that our approximation algorithms consistently arrive at about 99\% of the optimal. 
\end{abstract}

\section{Introduction}

Several online data marketplaces have emerged over the past years. We
as the database community have started to lay the foundations of such
data marketplaces and define and understand models of their
functionality \cite{Koutrisb, Li2012a, Tang2012}. The goal of this new
line of research has so far taken the form of how to design automated
algorithms for pricing query requests to a database. Most of the
current data marketplaces work with ad-hoc rules, mainly offering a
menu of queries and prices to the users or offering subscription based
services for whole datasets. The sudden increase of the amount of
online data and the different needs of the potential buyers will soon
render such an approach infeasible or inefficient in terms of both
revenue and social welfare of the system.

%This paper contributes to this new line of research. 
However, most of previous work has taken the axiomatic approach on
pricing queries to a relational database \cite{Koutrisb, Li2012a,
  Tang2012}. They posed a set of reasonable axioms like
arbitrage-freeness, no-disclosive pricing, maximality of prices and
then characterized the pricing functions that adhere to these axioms,
most of them showing uniqueness of prices.
%However, the axiomatic approach is rather ad-hoc and doesn't lay the foundations on why a pricing scheme should satisfy them.

However, none of the literature has given guidelines to the
foundational question: How should the market-maker actually price the
fundamental queries of her database? Most of the literature is focused
on how to derive prices given some fundamental input prices from the
seller, e.g., prices on basic queries or prices on cells of the
database, but has not addressed the question of how to compute these
initial prices in the first place!

In this work we resolve the above problem by using a game theoretic
approach to pricing relational queries building on recent ideas from
the optimal pricing literature
\cite{Guruswami,Balcan2008,Cai2011}. Using such techniques we actually
give an answer to the question above: What should be the optimal
revenue pricing scheme for the seller. As a first major result, we
show that several of the axioms assumed in the recent database
literature are implications of the optimal pricing approach thereby
justifying them and at the same time laying theoretical foundations
behind them.

%Unlike previous work we will also study selling exclusive access to data rather than simply pricing them. Therefore, we will also consider an auction formulation of the problem in addition to the pricing model. We will be interested in designing auctions for relational data that will maximize or approximately maximize the seller revenue.

Online data marketplaces, in addition to optimal, require simple
pricing schemes, with simple and intuitive rules. This has been the
motivating force behind the recent query-based pricing framework
\cite{Koutrisb}, that provides a simple and easy-to-compute pricing
approach. In contrast, most of the optimal mechanism design literature
concludes that the optimal mechanism that the seller should run
involves very complicated pricing rules. Our second major result in
this paper is to merge the two approaches of the database and the
optimal mechanism design community in the following way: \emph{We will
  restrict attention to the simple query-based pricing schemes of
  \cite{Koutrisb}, but we actually solve for the optimal pricing among
  these schemes.}

The main difference in our paper with previous work from the database
community is that we will also model the buyer preferences using
widely used utility theory models, i.e., we will assume that the users
(players) have \emph{valuations} over data in the relational database.

\paragraph*{Buyer Utility Model} 
We consider the following model: There are $n$ buyers and one seller
possessing information in the form of a relational database. Each
buyer $i$ is interested in the result of a query $\dq_i$ to the database
that is sold and has a value $v_i$ for the results of this query. The
query $\dq_i$ that each buyer is interested in, falls into some set of
queries $\dqset$ predefined by the seller. We will assume that
the pair $(v_i,\dq_i)$ of every buyer is known to the seller and hence,
the problem of maximizing revenue is an algorithmic question. In
practical applications one should think of the above set of buyers as
a representative demand that the seller has calculated will arrive in
his market. For instance, through market analysis he might have
concluded that out of a 100 buyers, $x$ of them will have value $v_1$
and will want query $\dq_1$, $y$ will have value $v_2$ and want query
$\dq_2$ and so on. We simply explicitly represent this representative
demand load as a set of buyers.

\paragraph*{Query Pricing Model} We restrict the pricing mechanism of the seller according to the query-based pricing framework of Koutris et al. \cite{Koutrisb}. Specifically, the seller sets explicitly the price to a bundle of queries $\bqset$, which are called the \textit{base queries}. From this 
set of queries, the price of every query in the database is computed according to the Fundamental Pricing Function introduced in \cite{Koutrisb}. 
%We will discuss more explicitly how the axioms used in \cite{Koutrisb} are simple implications
%of optimality of the pricing scheme, when restricted to query-based data pricing. 

We will assume that the set of base queries $\bqset$ and demanded queries $\dqset$ is such that the Fundamental Pricing 
Function is polynomially computable. For instance, if the set of bundles $\dqset$ that the buyers are interested in are Generalized Chain Queries, as described in \cite{Koutrisb} and the base-queries are all the selection queries, then \cite{Koutrisb} shows that the computation of the Fundamental Pricing Function can be reduced to a max-flow computation. 

When the seller chooses a set of base queries and their prices, then the price $p(\dq_i)$ of each bundle is uniquely determined. A buyer will buy the bundle if and only if $p(\dq_i)\leq v_i$. Hence, given a choice of input prices for the base queries the revenue of the seller is: 
\begin{equation}
\rev = \sum_{i\in [n]: v_i \geq p(\dq_i)} p(\dq_i)
\end{equation}
The above formulation allows us to ask the main question of this work:

\noindent\textbf{Main Question.} \emph{Given a set of bidder valuations and demanded queries what is the choice of input prices on base-queries that maximize the total revenue of the seller? }

\subsection*{Our Results}
We start by presenting a very strong connection between the query-based pricing model 
and the envy-free optimal pricing literature that has been recently developed by the algorithmic game theory community. Specifically, we show that the fundamental pricing formula of \cite{Koutrisb} is a consequence
of envy-freeness in envy-free pricing with appropriately defined buyer valuations, which we call \text{unit-bundle-demand valuations}. In addition, we show
that the polynomial computation of the pricing formula simply corresponds
to saying that the demand function (i.e. set that the player wants to buy given a set of prices) in the envy-free pricing model is polynomially computable, even when the valuation of the bidder is given in the concise
representation of a value $v_i$ and a query $Q_i$.

We use this strong connection to show that the optimal pricing question that we ask is NP-hard even in very natural instances that might arise in practice. We then proceed with presenting the logarithmic approximation, randomized single price-scheme of Balcan et al. \cite{Balcan2008} in our setting. We use the polynomial computability of the demand function to obtain a deterministic single pricing scheme that obtains a logarithmic approximation to the optimal pricing scheme.

Then we present novel multi-price pricing schemes that built upon the intuition obtained by the single-price 
schemes. We give two multi-price schemes that achieve revenue at least as much as the single-price
schemes pointwise, for each instance of the problem. The hardness of the pricing problem is 
mainly due to the fact that we don't know which subset of the players are going to be allocated 
by the optimal pricing. Our algorithms are based upon strategically picking a polynomial number of subsets of 
allocated players and then solving for the optimal or approximately optimal pricing, subject to each 
allocation set, outputting the set that yielded the highest revenue. When solving for the optimal pricing
conditional on the allocation set, we use the polynomial computability of the demand function
implied by the max-flow reduction of \cite{Koutrisb} as a separation oracle.

In the last part we provide some experimental analysis of how our new multi-price approximation schemes
perform with respect to the existing single-price schemes from the envy-free pricing literature and
show that they yield a substantial increase to the revenue produced, over random input instances. In addition, for instances where the exponential computation of the optimal
pricing is feasible we show that our approximation algorithms guarantee more than 99\% of the optimal.

\subsection*{Related Work} 

\textbf{Axiomatic Pricing of Relational Data.} Balazinska et al. \cite{Balazinska2011} motivated the need of new models
for capturing online market places. The stress the need for more fine-grained pricing and try to motivate
more automated pricing systems than the existing "price menu"-based ones. They mention several modeling
parameters that a pricing market should define: 1) the \textit{structural granularity} at which the 
seller should attach prices (this could either be tuples or cells of tables, queries, relations etc.)
2) a \textit{base pricing function} attaching seller defined prices at the chosen granularity 3) a \textit{method for determining the price} of every other allowable interaction with the database (e.g. query price) from the base pricing function, 4) a \textit{subscription model} specifying how updates to the data are priced. They also 
axiomatically define some desiderata for pricing functions, like arbitrage-freeness, fairness and efficiency. 

Koutris et al. \cite{Koutris, Koutrisb} considered a query based model of pricing. In their model, the structural
granularity of pricing is a query. The seller sets the price of a specific set of base queries and then the
automated algorithm has to derive the price of any other query that could be made to the database. 
They show that given a base pricing function there exists a unique maximal arbitrage-free pricing 
function and give a fundamental formula that defines it. They show that computing the price of a query 
based on this function is NP-hard for the case when base queries are general conjunctive queries. 
Their main result is a polynomial time algorithm, via a reduction to max-flow, for computing the price of any Generalized Chain Query (a special type of join conjunctive query) when the base queries are only selection
queries. The arbitrage property studied by \cite{Koutrisb} is based on an instance based determinacy 
relationship: if for the current instance of the database, one could determine the result of 
a query $Q$ be the result of a bundle of queries $Q_1,\ldots,Q_k$ then it must be that $p(Q)\leq \sum_{i=1}^{k}p(Q_i)$. In addition, in a follow-up paper \cite{Koutris2012a} they provide a description
of an implementation of their framework as a working marketplace. 

Li et al. \cite{Li2012a} studied the pricing of aggregate queries and more specifically the pricing
of the special type of linear queries (i.e. aggregate queries that can be expressed as a linear combination
of the column entries of a relationship). They also follow a similar approach to \cite{Koutrisb} in that
the seller prices a base set of linear queries and then the market algorithm computes the price of 
the remaining queries such that it satisfies arbitrage freeness. There are several crucial differences
with the approach of \cite{Koutrisb}. First, the arbitrage freeness is based not on an instance based
determinacy relationship but rather on a schema based one: if for any instance of the database, one
could determine a linear query $Q$ using the answer to a bundle of linear queries $Q_1,\ldots,Q_k$
then it must be that $p(Q)\leq \sum_{i=1}^{k}p(Q_i)$. This crucial difference makes the price function
independent of the instance of the database. This is a crucial property that Li et al. actually 
require as an axiom; if the price function were dependent on the instance then just asking for a price
quote for a specific query could potentially reveal information about the database entries. The latter
could potentially lead to manipulation by the buyers. This is the property of non-disclosive
pricing introduced in Liu et al.. It is interesting to observe that the pricing function of Koutris et al.
doesn't satisfy this property and hence is potentially susceptible to such manipulations. 

Tang et al. \cite{Tang2012} take a different approach than the previous two papers and consider tuples 
of relations as the structural granularity of the pricing function. Hence, the seller now places prices
on tuples of the database and then the price of a query is a function of the tuples that contribute to 
the response of the query. They also construct pricing functions that are arbitrage free. In addition 
they also apply their pricing model to probabilistic databases. 

Last, a very recent working paper \cite{Li2012b} tries to address the question of query pricing in combination with differential privacy. However, we don't plan to address differential privacy issues in this project
hence we just refer to it for completeness purposes. 

\textbf{Optimal Mechanism Design.} If we think of each query as an item then our setting could be cast as a combinatorial auction with single-minded bidders, where each query is an item and each bidder is interested in a specific bundle of items/queries and has a value $v_i$ for acquiring that bundle. In the Bayesian setting the type of each player is his bundle of interest and his value. Under such a formulation our question becomes the classical question of optimal pricing in a multidimensional setting. When there are constraints on the allocation of queries to buyers (i.e. if a serve this query to this buyer then I cannot serve this other query to some other buyer) then the problem becomes the classical problem of mutli-dimensional optimal mechanism 
design. Both problems have daunted the theoretical economics community for several decades since the seminal work
of Myerson \cite{Myerson1981} on the single-dimensional version of the problem, with only partial progress. Very recently 
the algorithmic game theory community has given algorithmic solutions to these problems. Specifically, Cai and Daskalakis \cite{Cai2011}
solves the multi-dimensional pricing problem, under conditions of the distribution of types. Cai, Daskalakis and Weinberg \cite{Cai2012,Cai2013}
solve the multi-dimensional mechanism design setting under very general feasibility constraint and in running time polynomial in the 
size of the type spaces. Recently, Daskalakis, Deckelbaum and Tzamos \cite{Daskalakis2012} show hardness results when the description
of the type distributions is given in a concise form and not explicitly. Another recent work by Bhalgat et al \cite{Bhalgat2012}, gives
a much simpler solution to the optimal mechanism design problem using a multiplicative weight updates approach. 

\textbf{Envy-Free Item Pricing.} Our formulation of the problem however is significantly different from the above. We don't consider the whole space of pricing schemes
but only a specific simple pricing scheme and optimize only over that. The reason being that most of the above techniques tend to give
very complicated pricing rules and also rules. 

Our formulation falls into the recent literature on envy-free pricing with unlimited supply, starting by Guruswami et al. \cite{Guruswami}. In envy-free pricing a seller faces a set of $n$ buyers. Each buyer has a valuation $v_i(S)$ for each set $S$ of items. The seller is restricted to assigning prices to the items. Given some item prices each buyer will pick the set that yields him the highest utility: $v_i(S)-\sum_{j\in S}p_j$. The literature then
asks the question of how should the seller price the items to maximize or approximately maximize his revenue.

Guruswami et al \cite{Guruswami} initiate the complexity study of this problem. They study the case where
bidders are either unit-demand or single-minded (i.e. want a specific set at some value $v_i$). They show that
in both cases the problem is NP-hard and in the second it is even hard to approximate, and give 
a $\log(n)+\log(m)$ approximation scheme for single-minded bidders that is based on posting
the same price on all items. Balcan et al. \cite{Balcan2008} extend
the latter approximation scheme to any combinatorial valuation using a randomized single price scheme. 
We derandomize the pricing scheme of Balcan et al. \cite{Balcan2008} for the specific valuation
classes that arise when buyers demand chain queries and the seller prices selection queries, to obtain
a deterministic single price scheme for such valuations that are a generalization of single-minded bidders.
On the hardness side, Demaine et al. \cite{Demaine} showed that under a mild complexity assumption 
no sublogarithmic ($O(\log^{\epsilon} n)$) polynomial approximation algorithm exists, hence rendering 
the latter approximation algorithms almost tight. Several other works have considered
special instances of the envy-free pricing framework both form approximation and complexity perspective
\cite{Chalersmook,Briest}.

\section{Query Pricing Model}
In this section we describe in more detail the practical query-based pricing model of Koutris et al. \cite{Koutrisb}. In the next section we will formulate our optimal pricing question under this model and discuss it's relation to the optimal item-pricing problem with single-minded bidders initially studied in \cite{Guruswami}.

Consider a seller that possesses an instance $D$ of a database under some relational schema $\rel=(\rel_1,\ldots,\rel_k)$ and let $\qset$ be the set of all queries to the database. For each relation $\rel$, denote with $\rel.X$ an attribute $X$ of that relation. The seller selects a set of base queries $\bqset$ and explicitly assigns a price $\bpr{\bq}$ for any base query $b\in \bqset$ and we denote with $\bprs$ the vector of input prices. Each potential buyer $i$ is interested in the results to a query $\dq_i$ on the database. The seller restricts this query to fall into some predefined set of demand queries $\dqset$. Given the set of explicitly priced queries, a price function $\pr{\dq}$ defines the price of any possible query to the database $\dq\in \qset$. 

The pricing function cannot be arbitrary but rather has to satisfy some natural axioms. Koutris et al. \cite{Koutrisb} define a minimal set of two axioms: arbitrage-freeness and discount-freeness, which we briefly describe below. 

The notion of \textit{arbitrage-freeness} states that if query $\dq$ can be ``determined'' by queries $\dq_1,\ldots,\dq_k$ then it must be that $\pr{\dq}\leq \sum_{i=1}^k \pr{\dq_i}$. Observe that 
the latter definition heavily depends on the notion of determinacy used (see \cite{Koutrisb} for
a detailed exposition of the properties and the different notions of determinacy). We will assume
here some abstract notion of determinacy and will denote with $D\vdash \cup_{i=1}^{k}\dq_i \twoheadrightarrow \dq$, if queries $\dq_1,\ldots,\dq_k$ determine $\dq$ in the database instance $D$.
A price function $\pr{\dq}$ is \textit{valid} if it is arbitrage-free and the price of explicitly priced queries is equal to the input price: $\forall \bq\in \bqset: \pr{\bq}=\bpr{\bq}$.

A pricing function is \textit{discount-free} if for any other valid pricing function $\prp{\dq}$
it holds: $\forall \dq\in \qset: \prp{\dq}\leq \pr{\dq}$, i.e. no other valid pricing function
assigns a higher price to any query in the database.

Koutris et al. \cite{Koutrisb} show that, given set of base queries and input prices, there is a unique pricing function that satisfies the axioms of arbitrage-freeness and discount freeness. Specifically, they characterize this pricing function as follows:

\begin{theorem}[Fundamental Pricing Function \cite{Koutrisb}]
 Let $\bqset$ be a set of base queries and $(\bpr{\bq})_{\bq\in \bqset}$ be the input prices.
 For any query $\dq\in \qset$, let 
\begin{equation}
supp_D^\bqset(\dq) = \{ \mathcal{C}\subseteq \bqset~|~ D\vdash \cup_{\bq_i\in \mathcal{C}} \bq_i \twoheadrightarrow \dq\} 
\end{equation}
be the set of support sets of query $\dq$, i.e. a set of base queries $\mathcal{C}\in supp_D^\bqset(\dq)$ if they determine $\dq$. There is a unique pricing function that satisfies the
axioms of arbitrage-freeness and discount-freeness and is defined as:
\begin{equation}\label{eqn:fund_price}
 \pr{\dq} = \min_{\mathcal{C}\in supp_D^\bqset(\dq)} \sum_{\bq \in \mathcal{C}} \bpr{\bq}
\end{equation}
\end{theorem}
The description of the function is pretty natural: consider all possible sets of queries among the explicitly priced queries, that determine query $\dq$. Then the price of query $\dq$ is the price of the
cheapest such set.

\paragraph*{Polynomial Computability of Price Function} Koutris et al \cite{Koutrisb} show that computing the price function is an NP-hard problem in general and characterize under which assumption
on the determinacy relation, on the set of base queries $\bqset$ and on the set of demand queries $\dqset$ the function is polynomially computable. 

\textbf{Polynomial Oracle Assumption.} In this work we will assume that the sets $\bqset$ and $\dqset$, as well as the determinacy relation are such that there exists a polynomial time algorithm $\orac{\dq}$ that computes the minimizer in the Fundamental Pricing formula \eqref{eqn:fund_price}: 
\begin{equation}
\orac{\dq} = \arg\min_{\mathcal{C}\in supp_D^\bqset(\dq)} \sum_{\bq \in \mathcal{C}} \bpr{\bq} 
\end{equation}
Our approximation algorithms will use oracle access to this algorithm. Specifically, our main algorithm
will use it as a separation oracle to solve in polynomial time a linear program with exponentially many
constraints.

The oracle assumption that we make is not a vacuous one and in fact Koutris et al. \cite{Koutrisb} show
that for a very natural class of base and demand queries the above minimizer can be computed via
a reduction to a min-cut computation in a query specific graph. More specifically, reinterpreting 
the results of \cite{Koutrisb}, suppose that the notion of instance-based determinacy is
used, the class of base queries is the set of all selection queries to the database and
the class of demand queries is a subset of conjunctive queries, called chain queries. 
Then for every query $\dq\in \dqset$ one can construct a weighted graph with two special vertices $s$ and $t$ in polynomial time, such that all the edges with finite weight correspond to base queries and have weight equal to their price and all other constructed edges have infinite weight and such that all finite weight $s$-$t$ cuts of the graph correspond to sets $\mathcal{C}\in supp_D^\bqset(\dq)$. 
We defer a more detailed exposition of the above class of queries for the experimental section
where we present experimental results for a simplified version of the above class of queries.

\section{Envy-Free Item Pricing}
In this section we re-interpret the analysis of query-pricing and cast it as an envy-free item pricing
problem with unlimited supply of items \cite{Guruswami}. This re-interpretation will allow us to 
use techniques and results from envy-free pricing as well as make the exposition much cleaner. In additional it will allow us to formulate the question of maximizing revenue in query-based pricing. 

In the setting of envy-free item-pricing, the seller possesses a set of items $J$ each one in infinite
supply. Each of the $n$ buyers has some combinatorial valuations $v_i(S)$ for getting a set of
items $S\subseteq J$. If a buyer gets a set of items $S\subseteq J$ and pays $p_j$ for each
item $j\in S$ then his utility is quasi-linear:
\begin{equation}
 u_i(S,{\bf p_S})=v_i(S) - \sum_{j\in S}p_j
\end{equation}

A vector of prices ${\bf p_J}=(p_j)_{j\in J}$ on the items and an allocation $S_i$ for bidder $i$ 
is \textit{envy-free} if no can increase his utility by selecting some other set at the given prices, i.e. he doesn't envy some other allocation. In other words the set $S_i$ has to satisfy:
\begin{equation}
 S_i = \arg\max_{S\subseteq I} u_i(S,{\bf p_S}) = \arg\max_{S\subseteq I} v_i(S) - \sum_{j\in S}p_j
\end{equation}
Given a vector of prices ${\bf p_J}$ and allocations $(S_i)_{i\in [n]}$ the revenue of the seller
is: $\sum_{i\in [n]} \sum_{j\in S_i}p_j$.

The envy-free pricing literature asks the question of computing the optimal pair of prices and allocations such that the total revenue is maximized, subject to the constraint that allocations 
should be envy-free. 

In the next section we show that if we pick an appropriate natural definition of bidder valuations
then the envy-free optimal pricing question becomes equivalent to our initial question of computing 
optimal input prices for base queries in the arbitrage-free pricing framework.

\subsection{Equivalence to Arbitrage-Free Pricing}

To even be able to formulate the question of optimal query pricing in the arbitrage 
free framework we need to introduce a buyer valuation model. Our valuation stems from the 
following natural assumption: we assume that each buyer has a value $v_i\in [0,H]$ (where $H$
is some upper bound on the valuation) for getting the 
results to his demand query $\dq_i$. Hence, a player gets a value $v_i$ if he gets the responses to query $\dq_i$ or to any other set of queries $\dq_1,\ldots,\dq_k$ that determine $\dq_i$.
Then our main question asks, how should the seller price his base queries such that 
he maximizes his total revenue, assuming that each buyer pays the price implied by the
Fundamental Pricing Function for his demanded query (as long as this price is below $v_i$).

For any such instance of the optimal query-pricing framework we present a corresponding instance
of the envy-free pricing framework that renders the above question equivalent to the optimal envy-free pricing question. We will interpret each of the queries $\bqset$ that the seller prices explicitly as items of unlimited supply in the envy-free pricing instance. 

Each bidder $i$ has a valuation $v_i(S)$ for getting a set of items $S$ defined as follows: he gets a value $v_i\in [0,H]$ if he gets the set of items corresponding to any set $\mathcal{C}\in supp_D^\bqset(\dq_i)$
or a superset of such a set and $0$ otherwise. More formally, the valuation of buyer $i$ for buying a set of items $S$ ss:
\begin{equation}
v_i(S) = \begin{cases}
v_i & \text{~if~} S\supseteq \mathcal{C} \in supp_D^\bqset(\dq_i)\\
0 & \text{~o.w~}
\end{cases}
\end{equation}
We will refer to such valuations as unit-bundle-demand valuations.

The seller simply sets prices on the items. Observe that for such valuations, given any vector of item prices, an allocation is envy-free only if each buyer is getting the set $\mathcal{C}_i\in supp_D^\bqset(\dq_i)$ that gives him his value $v_i$ at the lowest price:
\begin{equation}
\forall \mathcal{C}'\in supp_D^\bqset(\dq_i): \sum_{j\in \mathcal{C}_i} p_j \leq \sum_{j\in \mathcal{C}'}p_j
\end{equation}
Hence, given a set of item prices or equivalently base query prices, a buyer interested in a query $\dq_i$ will pay:
\begin{equation}
\min_{\mathcal{C}\in supp_D^\bqset(\dq_i)}\sum_{j\in \mathcal{C}} p_j,
\end{equation}
subject to the latter quantity being lower than his value. The latter is exactly the price
given by the Fundamental Pricing Formula \eqref{eqn:fund_price} for query $\dq_i$.

\emph{Thus maximizing revenue in the query-pricing model translates to finding the 
optimal envy-free pricing in an instance where bidders have unit-bundle-demand valuations.}

The other interesting aspect of the problem is that the valuations of the bidders is
given implicitly by the pair $(v_i,\dq_i)$. Hence, our algorithms should perform in polynomial
time with respect to this succinct representation of the input. The suncinctness of the representation
makes it potentially hard to compute the envy-free allocation. In essence, given the item prices
we don't know which set the player is going to choose, and hence we cannot even compute the 
revenue of some instantiation of the item prices. However, observe that computing the 
envy-free allocation of a bidder is equivalent to computing the minimizer in the Fundamental
Pricing Formula. Thus this corresponds exactly to the polynomial computability of the 
Fundamental Pricing Formula. Therefore, for instances where our polynomial oracle assumption holds
we can use the oracle $\orac{\dq}$ (e.g. min-cut for chain queries) to compute the envy-free allocation of a bidder, even under such succinct representation of the valuation. 

% The polynomial time algorithm of Koutris et al \cite{Koutrisb} for computing the pricing function
% translates to saying that given the item prices, the demand function of each bidder is polynomially 
% computable. Meaning, that given the item prices, we can polynomially compute which items the 
% bidder is going to buy, that maximize his utility. This is not always true for any combinatorial
% valuation and the polynomial computability heavily relies on the structure of input and demand queries
% used by \cite{Koutrisb}. Thus the same arguments show for the given type of valuations, computing
% the demand set of a bidder reduces to a max flow problem. 

In subsequent sections we will use this property of the valuation functions so as to derandomize
randomized pricing schemes proposed in the literature and as a separation oracle in an improved
pricing algorithm that we propose.

\paragraph*{Hardness of Optimal Pricing}
Even in the case when computing the envy-free allocation is polynomially computable, the problem of finding optimal envy-free prices has been shown to be NP-hard and even hard to approximate \cite{Guruswami}. For instance, the problem is NP-hard even when each buyer is interested only in a specific single set (i.e.  $supp_D^\bqset(\dq_i)$ is a singleton set) and has a value $v_i$ for acquiring it (single-minded bidders). In the appendix we give an alternative NP-hardness reduction that shows that the problem is hard even for instances that naturally arise from database query instances, when each relation has a single attribute and the queries of the bidders are either selections or a join the involves all relations.

\section{Approximately Optimal Pricing}
Given the hardness results in the literature mentioned in the previous section, in this section we address the problem of finding approximately optimal pricing schemes. We start by a simple random pricing scheme proposed by Balcan et al. \cite{Balcan2008} that assigns the same price on all base queries and that yields an $O(\log(n)+\log(m))$ approximation to the optimal revenue where $n$ is the number of buyers and $m$ is the number of items or equivalently, explicitly priced base queries (i.e. for instance all selection queries of the database). Then we show how to derandomize this pricing scheme to obtain more robust guarantees, by using
the polynomial oracle assumption.

\subsection{Random Price}
Balcan et al. \cite{Balcan2008} show that in the envy-free pricing model, for arbitrary valuation functions, as long as the value of the bidders are bounded by some constant $H$, then a simple randomized single-pricing scheme achieves a $O(\log(n)+\log(m))$ approximation to the optimal pricing scheme. We present this scheme in the context of query-based pricing.

\begin{algorithm}[h]
\SetKwInOut{Input}{Input}\SetKwInOut{Output}{Output}
\Input{$H = \max_{i} v_i$}
\BlankLine
\nl  Let $q_l = \frac{H}{2^{l-1}}$, for $l\in \{1,\ldots,\lfloor\log(2nm)\rfloor\}$\;
\nl Pick a random price $p$ uniformly at random from the set $\{q_1,\ldots,q_s\}$\;
\nl Price all base queries in $\bqset$ at price $p$.
\caption{Randomized Single Price Scheme.}
\end{algorithm}
\renewcommand{\algorithmcfname}{ALGORITHM}

Although the above price gives a nice worse case approximation guarantee, it only achieves it
in expectation. In the next section we show how to derandomize the above scheme to obtain
a deterministic single-price scheme that achieves the same guarantee deterministically. 

\subsection{Derandomizing Using Polynomial Oracle}
The analysis in this section is a generalization of the analysis in Guruswami et al. \cite{Guruswami}
that provide a similar deterministic single pricing scheme, only for the case of single minded
bidders and not for the more general unit-bundle-demand valuations that we need to cope with. 

Our analysis starts with the following observation: if the price of all base queries is the same
then the minimizer of the Fundamental Pricing Formula does not depend on the actual price. Rather 
it is simply going to be the support set $\mathcal{C}_i \in supp_D^\bqset(\dq_i)$ of minimum size. Let $t_i=|\mathcal{C}_i|$ be the size of that set. Given the demanded query $\dq_i$ of each buyer, the size of his minimum set and the minimum set itself is polynomially computable by $\orac{\dq_i}$. Simply place a price of $1$ on all base queries and invoke the oracle. 

Given the value $v_i$ and the size $t_i$ of his smallest support set, we assign
the same price to all base queries according to Algorithm \ref{alg:det-single}. We 
prove that this deterministic algorithm gives a $\log(n)+\log(m)$ approximation guarantee.
\begin{algorithm}[h]\label{alg:det-single}
\SetKwInOut{Input}{Input}\SetKwInOut{Output}{Output}
\Input{$(v_1,t_1),\ldots,(v_n,t_n)$}
\BlankLine
\nl Let $\pi_i = \frac{v_i}{t_i}$ be the per-base-query value of each buyer \;
\nl Reorder players such that $\pi_1\geq \pi_2\geq \ldots\geq \pi_n$\;
\nl For $i=\{1,\ldots,n\}$ let $\rev_i$ be the revenue obtained by pricing all base queries at
$\pi_i$. \;
\nl $\rev_i$ is polynomially computable since $t_i$ is and is simply $\pi_i\sum_{i'\leq i}  t_{i'}$\;
\nl Let $i^* = \arg\max_i \rev_i$. Price each selection query at $\pi_{i^*}$
\caption{Deterministic Single Price Scheme.}
\end{algorithm}
\renewcommand{\algorithmcfname}{ALGORITHM}

\begin{theorem} Algorithm \ref{alg:det-single} computes a pricing that is $O(\log(n)+\log(m))$-approximately
optimal.\end{theorem}
\begin{proof}
If all items are charged at $\pi_i$ then $\rev_i = \pi_i\sum_{i'\leq i}t_{i'}$. Since the revenue $\rev$ of the algorithm
is at least $\rev_i$ for all $i$ we get: $\rev \geq \frac{v_i}{t_i}\sum_{i'\leq i}t_{i'}$. Hence, 
\begin{align}
v_i \leq \rev\frac{t_i}{\sum_{i'\leq i}t_{i'}}
\end{align}
Summing over all players we get:
\begin{align*}
\sum_i v_i \leq~& \rev \sum_i \frac{t_i}{\sum_{i'\leq i}t_{i'}}
=~\rev \sum_i \sum_{k=1}^{t_i} \frac{1}{\sum_{i'\leq i}t_{i'}}\\
\leq~& \rev \sum_i \sum_{k=1}^{t_i} \frac{1}{k+\sum_{i'< i}t_{i'}}
=~  \rev\sum_{k=1}^{\sum_i t_i} \frac{1}{k}\\
\leq~& \rev\sum_{k=1}^{n\cdot m} \frac{1}{k} = \rev \cdot H_{n\cdot m} \leq
\rev\log(n\cdot m)
\end{align*}
Now the theorem follows by simply observing that the revenue of the optimal
pricing scheme is at most $\sum_i v_i$ since we can extract a revenue of at 
most $v_i$ from each bidder.
\end{proof}

In fact the above theorem shows that this pricing scheme achieves at least a
$\log(n\cdot m)$ bound with respect to the total value of the bidders, which is a 
stronger benchmark than the optimal revenue. In fact, the following easy example
shows that even when there is only a single base query to be priced no query-pricing scheme (even
ignoring computational constraints) can achieve a revenue better than $\log(n)$ of the total value 
of the buyers.

\begin{example}
Consider a set of $n$ buyers and a database with only a single priced base query. Suppose
that the value of buyer $i$ is $1/i$. Hence, the total value of the buyers is
$H_n$. On the other hand it is easy to see that any base query price will yield
a revenue of at most $1$. Observe that an optimal price will be of the form $1/k$
for some $k\in[1,n]$. However, if we post a price of $1/k$ then we know that only 
$k$ buyers have a valuation higher than $1/k$ and therefore only $k$ buyers will
ever buy. Hence, the total revenue will be $1$.
\end{example}

Last, it is interesting to note that the above pricing scheme is a very easy to implement and announce for large web applications. The seller simply needs to announce a single price to the potential
buyers and then each buyers price will depend on the information size of his query as is 
implied by the number of base queries needed to determine it. The latter is also an 
easy to describe explanation to the buyer for the price he had to pay for his query. 

%
%Since in the above example there was only one query to be priced, the inability to extract
%more than an $H_n$ factor of the total value is not based on the fact that we restricted to single
%price schemes. In the next example we show that for the restricted single-price schemes
%the $H_{n\cdot m}$ factor is tight. However, this still allows for non-single price schemes
%to yield more than an $H_{n\cdot m}$ fraction of the value. Multi-price schemes are examined
%in the next section.
%
%\begin{example}
%Consider a set of $n$ buyers and a database with $m$ priced queries. Each buyer is
%interested in the whole set of priced queries (i.e. his support set consists of the whole
%set of selection queries). Buyer $i$ has a value of $\sum_{j=1}^{m} \frac{1}{j\cdot n+i}$.
%Thus the total value of the buyers is $\sum_{i=1}^{n}\sum_{j=1}^{m} \frac{1}{j\cdot n+i}=H_{n\cdot m}$.
%Suppose that we post a single price $p$ for all the items. Observe that if $p>\frac{1}{m\cdot n+i}$
%then buyer $i$ will never buy.
%\end{example}

\section{Improved Multi-Price Schemes}
The pricing schemes in the previous section though approximately optimal in the worse
case, they take minimal advantage of the structure of the valuations of the bidders. 
It is obvious that an optimal pricing scheme would take advantage of the specific structure
of the bidder valuations and the demanded queries to assign different prices to different selection queries. 

In this section we present an algorithm that although asymptotically has the same worse
case guarantees, nevertheless performs strictly better than the single price schemes of the previous section point-wise, for each instance of the pricing problem. 

We first point out that when the base query prices are not the same then the minimizer of 
the Fundamental Pricing Formula for any query depends heavily on the actual prices and hence we cannot abstract by saying that it is always going to be the support set with the smallest size. 

Hence, the base query set that determines the price of a buyers query is affected by the input prices and the revenue is subsequently affected by this set. This creates a feedback problem that would lead to an exponential algorithm if we were to proceed in optimizing over the space of all possible input prices.

In addition, it is not clear which of the buyers the optimal pricing scheme serves. In fact the NP-hardness of the problem stems mainly from this fact. We portray this
in the next section by first considering the case of single-minded bidders. We propose an
improved algorithm for single-minded bidders and then we generalize it to the general class
of valuations that we cope with.

\subsection{Single-Minded Valuations}\label{sec:single-mind-multi}
In order to develop some intuition, in this section we consider the special case where the support set $supp_D^\bqset(\dq_i)$ of each buyer is a singleton. Thus each buyer is satisfied only by a specific
set of base queries $\mathcal{C}_i$. 

We start by observing that if we knew that the optimal pricing scheme was serving a set $N^*$ of bidders
then computing the optimal prices can be done in polynomial time. Specifically, conditional on the 
service set $N^*$ the optimal pricing scheme can be written as a packing LP:
\begin{align}
\rev(N^*)=~&\max_{{\bf p_{\bqset}}} \sum_{i\in N^*} \sum_{\bq\in \mathcal{C}_i}\bpr{\bq}& \nonumber\\
s.t.~& \sum_{\bq\in \mathcal{C}_i}\bpr{\bq} \leq v_i ~& \forall i\in N^* \label{lp:serve}\\
~& \bpr{\bq} \geq 0 ~& \forall \bq\in \bqset \nonumber
\end{align}

Thus a trivial exponential algorithm would be to enumerate over all $2^n$ possible service
sets, compute the revenue of the optimal pricing and then pick the service set that yields the
highest revenue. This is too costly and our approach is to identify a polynomial
number of candidate service sets to check which will guarantee that an approximately good pricing scheme
would be found in some of this polynomially sized search space.

The single-price schemes in the previous section and their analysis, gives us exactly that. Essentially
they show that in order to find a pricing scheme that is a logarithmic approximation to the optimal
pricing, one needs to consider only $n$ possible service sets: the sets of the form $\{1,\ldots,i\}$
for $i\in \{1,\ldots,n\}$, where players are ordered with decreasing per-item value $\pi_i = \frac{v_i}{|\mathcal{C}_i|}$.

This yields the following pricing scheme for single-minded bidders.
\begin{algorithm}[h]\label{alg:det-multi}
\SetKwInOut{Input}{Input}\SetKwInOut{Output}{Output}
\Input{$(v_1,\mathcal{C}_1),\ldots,(v_n,\mathcal{C}_n)$}
\BlankLine
\nl Let $\pi_i = \frac{v_i}{|\mathcal{C}_i|}$ be the per-base-query value of each bidder. \;
\nl Reorder players such that $\pi_1\geq \pi_2\geq \ldots\geq \pi_n$\;
\nl For $i=\{1,\ldots,n\}$ let $N_i=\{1,\ldots,i\}$\;
\nl Compute $\rev_i=\rev(N_i)$ obtained by LP \eqref{lp:serve}\;
\nl Let $i^* = \arg\max_i \rev_i$. Output the pricing obtained by solving LP $\rev(N_{i^*})$.
\caption{Multi-Price Scheme for Single-Minded Bidders.}
\end{algorithm}
\renewcommand{\algorithmcfname}{ALGORITHM}

\begin{theorem} Algorithm \ref{alg:det-multi} is $O(\log(n)+\log(m))$-approximately
optimal and is at least as good as Algorithm \ref{alg:det-single} for any instance of 
the pricing problem.\end{theorem}

In addition it is easy to see with the following example that the multi-pricing scheme 
presented above can actually achieve revenue $H_m$ times larger than the single-price scheme
of the previous section, where $m$ is the number of base queries.

\begin{example}
Consider a setting with $m$ buyers and $m$ selection queries. Each buyer is single-minded
and buyer $i$ wants only query $i$, while having value $v_i=1/i$ for it. 
Any optimal single price scheme will post some price $p=1/k$ on all the items
for some $k$. Any such single price will obtain a revenue of $1$ since only $k$
buyers will actually buy. On the other hand the multi-price algorithm will actually
price item $i$ with $p_i=1/i$ and achieve the whole value of $H_m$ as revenue.
\end{example}

\subsection{Unit-Bundle-Demand Valuations}
In this section we extend Algorithm \ref{alg:det-multi} to the general class of
unit-bundle-demand valuations that we cope with. For more general valuations, a bidder could be satisfied by more than one set of base queries. Hence, we also need to take care of arbitrage opportunities that the bidder might have and augment the LP of the previous section with no-arbitrage
or equivalently envy-free constraints. 

However, we also don't know the base set that is the minimizer of the fundamental pricing formula for each player, since that set depends on the actual prices. Handling the multiplicity of the interest sets would lead to a non linear constraint if we were to directly adapt the LP in the previous section. 

Hence, we need in some sense we need to fix the minimizer of the fundamental pricing formula for 
each player before optimizing over the revenue, i.e. we have to predetermine from which of his multiple sets the player is going to be satisfied by. Then we will formulate an LP that captures the fact that the player should not be envying being satisfied by any other of the sets in his support at the current price levels.

We will again use the intuition behind the single pricing schemes to determine the latter pre-determined set. Specifically, we will fix that each bidder will pick the set of queries from his support set of
minimum size. Observe that this was the set that he picked when we were setting the same price for all the items. Hence, it suffices to restrict ourselves to such a constraint, so as to get a logarithmic approximation. 

For every player $i$ let $\mathcal{C}_i$ be his support set of minimum size. Given a set of players $N^*$ that have to be served and under the constraint that the minimizer of the fundamental pricing formula has to be $\mathcal{C}_i$ for each player $i$, the optimal revenue can be solved by the following LP:
\begin{align}
\rev(N^*)&=~\max_{{\bf \bpr{\bqset}}} \sum_{i\in N^*}\sum_{\bq\in \mathcal{C}_i} \bpr{\bq} & \nonumber\\
s.t.&~ \sum_{\bq\in \mathcal{C}_i} \bpr{\bq} \leq v_i ~& \forall i\in N^* \label{lp:serve-2}\\
&~ \sum_{\bq\in \mathcal{C}_i} \bpr{\bq}  \leq \sum_{\bq\in \mathcal{C}_i'}\bpr{\bq} ~& \forall i\in N^*, \mathcal{C}_i'\in supp_D^\bqset(\dq_i)\nonumber\\
&~ \bpr{\bq} \geq 0 ~& \forall \bq\in \bqset\nonumber
\end{align}

The third set of constraints captures exactly the arbitrage-freeness or equivalently
the envy-freeness for each player $i$. At first glance, the latter LP might be feasible, and 
specifically there is no guarantee that it will have a feasible solution if we fix some arbitrary
predefined sets, to be the minimizers of the pricing formulas. However, observe that for the specific choice of predetermined sets that we did, the LP is always feasible for any $N_i=\{1,\ldots,i\}$ (assuming players are in decreasing $\pi_i$ order), since setting the price of every base query to $\pi_i=\frac{v_i}{|\mathcal{C}_i|}$ is a feasible solution. 

The latter also implies that the value $\rev(N_i)$ of the above LP is at least the revenue of the single pricing scheme in the previous section for each $N_i$. Therefore, if we consider all the serving sets of the form $N_i$ and then take the serving set that gave the maximum revenue, then we will get at least as much revenue as the single pricing schemes of the previous section.

The main problem with LP \eqref{lp:serve-2} is that the set $supp_D^\bqset(\dq_i)$ could potentially contain exponentially many elements. Hence, the above LP can have exponentially many constraints.

However, the polynomial oracle assumption provides a separation oracle for the above LP and hence it can be solved in polynomial time using the ellipsoid method. Specifically, given a set of prices $\bprs$, we can check whether the exponentially many constraints $\sum_{\bq\in \mathcal{C}_i} \bpr{\bq} \leq \sum_{\bq\in \mathcal{C}_i'}\bpr{\bq}$ are violated, by calling 
oracle $\orac{\dq_i}$ with input prices $\bprs$ (e.g. this would correspond to a call to the min-cut computation described in \cite{Koutrisb} for the case of chain queries). If the result of $\orac{\dq_i}$ is below $\sum_{\bq\in \mathcal{C}_i} \bpr{\bq}$, then some of the constraints
is violated. More specifically it is the constraint that corresponds to the set $\mathcal{C}_i'$ returned by the oracle as the actual minimizer for the given prices. Hence, running $\orac{\dq_i}$ for
each $i\in [n]$ serves as a polynomial separation oracle for LP \eqref{lp:serve-2}.

\subsection{A Simpler Faster Multi-Price Scheme}\label{sec:multi-comb}
The pricing scheme described in the previous section involves solving $n$ linear
programs. The computational overhead of this addition could be substantial. Hence,
it would be interesting to consider more combinatorial pricing algorithms that still 
are improvements to the single pricing scheme.

In this section we propose such a pricing algorithm, which is an alternative to solving 
LP \eqref{lp:serve-2} for each of the $n$ serving sets considered. The algorithm
that we propose is described in Figure \ref{alg:alg-2}.

\begin{algorithm}[h]\label{alg:alg-2}
\SetKwInOut{Input}{Input}\SetKwInOut{Output}{Output}
\Input{$(v_1,Q_1),\ldots,(v_n,Q_n)$}
\BlankLine
\nl Compute minimum size support set $\mathcal{C}_i$ for each buyer, by running the corresponding
oracle with input prices initialized to $1$.\;
\nl Let $\pi_i = \frac{v_i}{|\mathcal{C}_i|}$ be the per-base-query value of each buyer. \;
\nl Reorder players such that $\pi_1\geq \pi_2\geq \ldots\geq \pi_n$\;
\nl For $i=\{1,\ldots,n\}$ let $N_i=\{1,\ldots,i\}$ and compute $\rev_i=\rev(N_i,(v_1,\mathcal{C}_1),\ldots,(v_n,\mathcal{C}_n))$ using Procedure \ref{proc:compr}\;
\nl  Let $i^* = \arg\max_i \rev_i$. Output the pricing obtained for service set $N_{i^*}$.
\caption{Simpler Multi-Price Scheme for General Chain Query Demands.}
\end{algorithm}

\renewcommand{\algorithmcfname}{PROCEDURE}
\begin{algorithm}[h]\label{proc:compr}
\SetKwInOut{Input}{Input}\SetKwInOut{Output}{Output}
\Input{$(v_1,C_1),\ldots,(v_n,C_n)$ and service set $N_i$.}
\BlankLine
\nl Initialize: $S^0 = \bqset$, $N^0=N_i$ and $(v_i^0,\mathcal{C}_i^0) = (v_i,\mathcal{C}_i)$, $\forall \bq\in \bqset: \bpr{\bq}=\infty$\;
\While{$S^t = \cup_{i\in N^t}\mathcal{C}_i\neq \emptyset$}{
\nl Let $p^t = \min_{i\in N^t}\frac{v_i^t}{|\mathcal{C}_i^t|}$ and $i^t= \arg\min_{i\in N^t}\frac{v_i^t}{|\mathcal{C}_i^t|}$.\;
\nl Set $\bpr{\bq} = p^t$ for all base queries $\bq\in \mathcal{C}_{i^t}$.\;
\nl Set $\mathcal{C}_i^{t+1} = \mathcal{C}_i^t - \mathcal{C}_{i^t}^t$ and $v_i^{t+1} = v_i^{t} - |\mathcal{C}_i^t\cap \mathcal{C}_{i^t}^t|\cdot p$ for all $i\in N^t$.\;
\nl $N^{t+1} = N^t-\{i^t\}$
}
\Output{Return revenue produced by prices $\bprs$.}
\caption{Procedure for finding prices conditional on a service set $N_i$.}
\end{algorithm}
\renewcommand{\algorithmcfname}{ALGORITHM}

Algorithm \ref{alg:alg-2} replaces the computationally costly LP with a simpler algorithm that is not as optimal as solving the LP, but still always performs at least as good as the single pricing scheme.
In fact our experimental analysis shows that it performs almost as good as the LP algorithm of the previous section.

The alternative process described above can be easily summarized as follows: at each iteration pick the
player that has the smallest per-base-query valuation $\pi_i$ for the items that haven't already been priced. Then assign this per-base-query valuation as a price on all the currently unpriced base queries in the minimum size support set $\mathcal{C}_i$ of this player. Then update the values of the players for the unpriced base queries and update their minimum size support sets to be the unpriced base queries of their initial set. 

In the next theorem we formally show that the above process always produces a revenue at least 
as high as the single price revenue for any instance of the problem. Observe that under the prices
set by the latter algorithm it is not necessarily true that the minimizer of the fundamental pricing
formula is $\mathcal{C}_i$. However, despite this fact we still show that the revenue is always at
least as good as the single price revenue. We show that by showing that the revenue is always greater
conditional on any service set $N_i=\{1,\ldots,i\}$.

\begin{theorem} For any $N_i =\{1,\ldots,i\}$, Procedure \ref{proc:compr} produces revenue at least as much as posting
a single price of $\pi_i=\frac{v_i}{|\mathcal{C}_i|} = \min_{k\in N_i}\frac{v_k}{|\mathcal{C}_k|}$.
\end{theorem}
\begin{proof}
We show that the prices $\bprs$ output by Procedure \ref{proc:compr} are greater than or equal to $\pi_i$. Given the above observe that the total revenue must be greater since: the same set of bidders
$N_i$ are allocated by the single pricing scheme and by procedure \ref{proc:compr} and the minimizer 
of the pricing formula for each buyer contains at least as many base queries as that of the minimum cardinality set $\mathcal{C}_i$. Observe that in the pricing output by procedure \ref{proc:compr} it is not necessarily true that under such pricing the minimizer is going to be the minimum cardinality set, since prices are not the same for all items. Hence, a player's minimizing set might be some other support set that has a larger size, but is priced cheaper. However, our reasoning above still works, to produce revenue at least as much.

To show our initial and main claim, we first show that the per-base-query valuations of the players are monotonically increasing as a function of the step $t$,(assume that if a player is selected at some point in the process then after that his per-item value is infinity). We show this by induction. More, formally, let $\pi_k^t = \frac{v_k^t}{|\mathcal{C}_k^t|}$ be the per-base-query valuation of player $k\in N_i$ at step $t$. We will show that $\pi_k^t\geq \pi_k^{t+1}$.

By the way the algorithm works, we know that for any buyer: $\pi_i^t \geq \pi_{i^t}^t = p^t$. Hence:
\begin{align*}
\pi_k^{t+1} =~& \frac{v_k^t-p^t|\mathcal{C}_k^t\cap \mathcal{C}_{i^t}^t|}{|\mathcal{C}_k^t|-|\mathcal{C}_k^t\cap \mathcal{C}_{i^t}^t|}\geq \frac{v_k^t-\pi_k^t|\mathcal{C}_k^t\cap \mathcal{C}_{i^t}^t|}{|\mathcal{C}_k^t|-|\mathcal{C}_k^t\cap \mathcal{C}_{i^t}^t|}\\
=~& \frac{v_k^t-\frac{v_k^t}{|\mathcal{C}_k^t|}|\mathcal{C}_k^t\cap \mathcal{C}_{i^t}^t|}{|\mathcal{C}_k^t|-|\mathcal{C}_k^t\cap \mathcal{C}_{i^t}^t|} = 
\frac{v_k^t}{|\mathcal{C}_k^t|}\frac{|\mathcal{C}_k^t|-|\mathcal{C}_k^t\cap \mathcal{C}_{i^t}^t|}{|\mathcal{C}_k^t|-|\mathcal{C}_k^t\cap \mathcal{C}_{i^t}^t|} = \pi_k^t
\end{align*}

Now observe that by the assumption on the ordering of the players we have: $\pi_k\geq \pi_i$ 
for all $k\in N_i =\{1,\ldots,i\}$. Thus by the monotonicity of per-base-query valuations
we get that $\pi_k^t \geq \pi_i$ for any player $k$ and any step $t$.
Consider a base query $\bq\in \bqset$. Suppose that query $\bq$ was priced at iteration $t$ (if a query was never priced it means that it is not in the set $\mathcal{C}_i$ of any buyer and hence we can ignore it since it remains priced at a huge value $H$). Then we get: $p_\bq = p^t = \pi_{i^t}^t\geq \pi_i$.
\end{proof}

\section{Experimental Analysis}
In the previous section we presented new pricing algorithms that pointwise improve the existing approximation algorithms in the envy-free pricing literature, for the specific class of unit-bundle-demand valuations that arise from the query-pricing framework. However, our algorithms don't improve the asymptotic efficiency which theoretically still remains $O(\log(m)+\log(n))$. 

In order to justify the usability of these new algorithms we perform an experimental evaluation
of their performance when compared to the existing single-pricing schemes. In addition, we also compare them with respect to the exponential optimal algorithm for small instances, showing that it achieves an almost optimal performance. 

We first present our experimental results for the general class of unit-bundle-demand valuations
that arise when the base queries are all the selection queries in the database and the 
demand queries are the set of chain join queries described by Koutris et al. \cite{Koutrisb}.
For this general class of complex valuations we compare the efficiency of our fast multi-price scheme
with respect to the single price one. 

Subsequently, we focus on the simpler class of single-minded valuations, so as to portray how our algorithms perform with respect to the optimal pricing.

\subsection{Chain Query Valuations}
For this experimental evaluation we considered the instantiation of the query-pricing framework 
where the notion of instance based determinacy is used. The reader is redirect to \cite{Koutrisb} for 
more detailed exposition. In the instance, based determinacy framework each attribute is assumed
to take values from some predefined discrete set. Specifically, for each attribute $R.X$ of relation 
$R$, we denote with $Col_{R.X}=\{a_1,\ldots,a_t\}$ a finite set of possible values that attribute $R.X$ has to lie in. Thus the database has only a finite set of possible instances. Instance-based determinacy states that a set of queries $\dq_1,\ldots,\dq_k$ determine a query $\dq$ under instance $D$, if for any other possible instance $D'$, if the results of queries $\dq_1,\ldots,\dq_k$ remain exactly the same as in $D$ then it must be that the result of $\dq$ also remains the same.

Denote with $\sigma_{R.X=a_i}$ the selection query that corresponds to selecting all entries of relation $R$ that have value $a_i$ at attribute $R.X$. In addition denote with $\Sigma_{R.X} = \{\sigma_{R.X=a_i}~|~a_i\in Col_{R.X}\} $ to
be the set of a selection queries for a particular attribute of relation $R$, and denote 
with $\Sigma$ the set of all selections for all relations and all attributes. We consider the setting
where the set of base queries $\bqset$ is equal to the set of all selection queries $\Sigma$.

For simplicity of exposition we consider only relational schemas where each relation is a binary relation and each buyer is interested in a Chain Query as is described in \cite{Koutrisb}:
\begin{defn}A chain query is of the form of a join of $k$ relations: $\dq=R_1,\ldots,R_k$, where 
each $R_i$ is a binary relation and each relation shares exactly one variable with each successor.
\end{defn}
The latter setting is a special case of the Generalized Chain Queries of \cite{Koutrisb}, hence for 
the latter setting we know by \cite{Koutrisb} that the Fundamental Pricing Formula is polynomially
computable and is reduced to a min-cut computation. 

We implemented Algorithm \ref{alg:alg-2} instantiated for the latter setting and compared each revenue 
performance with respect to the single-price scheme, for varying parameters of the size of the schema,
the number of possible values of each attribute and the number of buyers. We observed that the revenue ratio increases significantly as a function of the number of binary relations on the attributes in the schema. Our instances where produced by sampling random chain queries on given relational schemas of varying size. The number of possible values of each attribute was held fixed. In Figure \ref{fig:alg-2-comp} we present how the revenue ratio varies with the number of relations. On the other hand we observed that the ratio remains constant with the number of possible values of the attributes. The main reason is that the base queries corresponding to selecting a specific value of some attribute of a binary relation in general are either all contained in the fundamental formula minimizer or none. Hence, they mostly act as a single item and therefore don't increase the complexity of the instances. We also observed that due to the random generation of our instances, as we increased the number of bidders by 
much, the ratio actually decreased. The main reason is that as the number of bidders becomes large all base queries affect almost homogeneously the revenue and hence a single price scheme is close to optimal. 

\begin{figure}
\centering
\includegraphics[scale=.6]{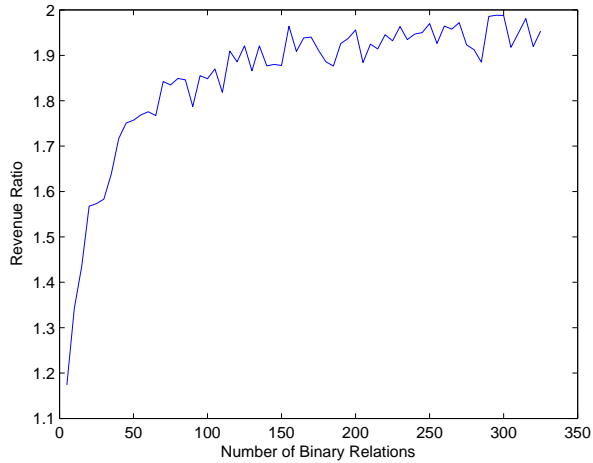}
\caption{This plot portrays the average ratio of the revenue obtained by the LP based multi-price
scheme of Section \ref{sec:single-mind-multi} over the single-price revenue.
The number of buyers is constant $20$ while the number of binary relations in the database varies from $5$ to $325$. The results imply that the ratio between the two algorithms might be increasing logarithmically in the number of binary relations for such random instances.}
\label{fig:alg-2-comp}
\end{figure}

\subsection{Single Minded Valuations}
To make the comparison with the exponential optimal algorithm feasible, in this section we focus on single-minded bidders, which corresponds to the case when the query of each bidder has a singleton support set, i.e. there is a unique set of base queries that determine it. 

We perform experiments where the bidders have a value drawn uniformly at random from some set. In addition their interest set is drawn uniformly at random as follows: we pick a random number $t\in [1,|\bqset|]$ and then we pick $t$ base queries at random with replacement. 

In Figure \ref{fig:varying_m} we compare the revenue performance of the two algorithms multi-price algorithms with resepct to the single price one as a function of the number of base queries. In Figure \ref{fig:opt} we compare our LP based multi-price algorithm \ref{alg:det-multi} to the optimal one.
We observe that the multi-price scheme consistently achieves at least 99\% of the optimal revenue.

\begin{figure}
\includegraphics[scale=.6]{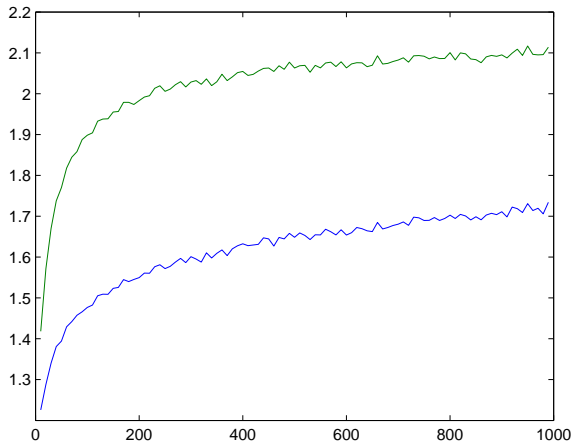}
\caption{This plot portrays the average ratio of the revenue obtained by the multi-price
schemes over the revenue of the single derandomized scheme. The input instances have $20$
buyers and the plot shows the average ratio as a function of the number of items/base queries 
$m$, varying from $10$ to $1000$. For each $m$, $1000$ random instances where ran and the average
ratio is depicted. The green line portrays the ratio of the LP based multi-price Algorithm \label{alg:det-multi} and the blue line the ratio of the combinatorial
multi-price scheme of Section \ref{sec:multi-comb}.}
 \label{fig:varying_m}
\end{figure}

% \begin{figure}
% \includegraphics[scale=.6]{varying_n.eps}
% \caption{This plot portrays the average ratio of the revenue obtained by the multi-price
% scheme over the revenue of the single derandomized scheme, when the number of items/base queries remains
% fixed at $50$ and the number of buyers ranges from $10$ to $600$.}
% \end{figure}

\begin{figure}
\includegraphics[scale=.6]{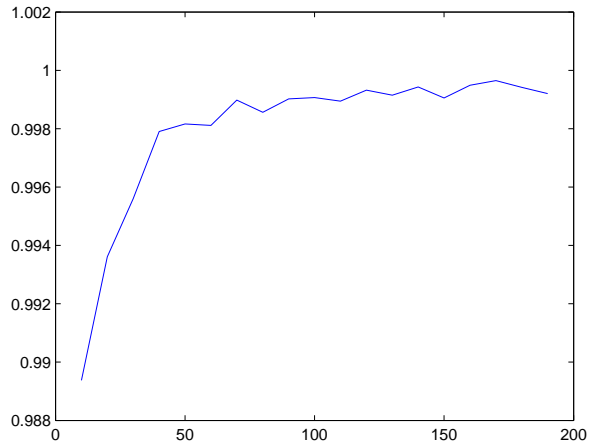}
\caption{This plot portrays the average ratio of the revenue obtained by the LP based multi-price
scheme of Section \ref{sec:single-mind-multi} over the optimal revenue obtained by an exponential algorithm.
The number of buyers is constant $7$ while the number of items/base queries varies from $10$ to
$200$.}
\label{fig:opt}
\end{figure}

\section{Discussion and Conclusions}
We asked the question of how to optimally set the price to the base queries in the query-based pricing framework introduced by Koutris et al. \cite{Koutrisb}. We showed that this problem is equivalent to an item pricing problem in the envy-free pricing framework introduced by Guruswami et al. \cite{Guruswami}. We gave a strong analogy between the arbitrage freeness axiom in the query pricing framework with the envy-freeness concept in the envy-free pricing framework.

We showed that the restricted setting studied by Koutris et al.\cite{Koutrisb} that allows for the polynomial computation of the pricing function, yields a valuation class in the envy-free pricing framework that has a polynomially computable demand function: given item prices one can compute in polynomial time the optimal set for each player and the total price that the player will pay for his optimal set. 

We used this structural property to improve on existing pricing schemes from the envy-free literature in two directions. First we used it to derandomize a random pricing guarantee so as to achieve more robust revenue approximations. Secondly, we considered a multi-price extension of the single-pricing schemes of the literature that always yield higher revenue pointwise for each instance of the pricing problem. 

Last we gave some experimental results showing that our new algorithm can give significant improvement on the revenue achievable by the single pricing schemes and that for small instances, where the optimal is computable in reasonable time, the outcome of our multipricing scheme is very close to optimal.

\section{Future Directions}
An interesting future direction is to apply our results to a Bayesian setting of incomplete information where the pair $(v_i,Q_i)$ of each bidder is independently and identically distributed according to some distribution $D$ on $\R_+\times 2^{Q}$ (Bayesian Setting). In the latter setting the auctioneer wants to set prices that maximize his expected revenue. One could use a discretization of the distribution where all probabilities are multiples of $\epsilon$. Then convert the bayesian problem into a complete information problem where each probability mass of $\epsilon$ is represented by a different buyer in the complete information setting, yielding an instance with $1/\epsilon$ buyers. It is reasonable that an approximate solution in this reduced instance will be an approximate solution to the initial pricing problem in the Bayesian setting, however a more rigorous analysis of such an extension is needed.

Another interesting future direction is try to reduce the computational overhead introduced by having
to solve linear programs in our multi-pricing schemes, while at the same time yielding very comparable revenue. We gave one such algorithm, but we believe there is room for improvement in that direction. It is possible that for this specially structured LPs that we formulated one could also find an equivalent combinatorial algorithm for finding an optimal or approximately optimal solution. An initial look towards that direction shows that LP \eqref{lp:serve} formulated for single minded bidders is close to an optimization over a polymatroid in the special case when the value of each bidder is uniquely defined by his interest set and is a submodular function of the interest set: i.e. the value of a player interested in a support set $\mathcal{C}_i$ is $v_i=v(\mathcal{C}_i)$, where $v(\cdot)$ is a submodular function. The only difference is that a subset of the constraints of the polymatroid has been dropped. Despite this fact, under natural assumptions a greedy approach of sequentially picking the item that is in the interest set of the most players and raising it's price as much as possible, would be a good approximation to the LP's value.

\paragraph*{Acknowledgments} The authors would like to thank \'{E}va Tardos and Renato Paes Leme for  fruitful conversations on the problem.

\bibliographystyle{abbrv}
\bibliography{pricing_queries}

\end{document}